\documentclass[final]{article}
\usepackage{graphicx, cite}

\usepackage{url}
\usepackage{graphicx}
\usepackage{subfigure}
\usepackage{amssymb}
\usepackage{newproof}

\begin{document}

\title{Synthesis of Reversible Functions\\Beyond Gate Count and Quantum Cost}
\author{
Robert Wille$^\dag$, Mehdi Saeedi$^\S$, Rolf Drechsler$^\flat$ \\ \\
\small $^\dag$ Institute of Computer Science, University of Bremen, Bremen, Germany\\
\small E-mail: rwille@informatik.uni-bremen.de \\\\
\small $^\S$ Computer Engineering Department, Amirkabir University of Technology, Tehran, Iran\\
\small E-mail: msaeedi@aut.ac.ir \\\\
\small $^\flat$ Institute of Computer Science, University of Bremen, Bremen, Germany\\
\small E-mail: drechsler@uni-bremen.de \\
}
\date{}

\maketitle

\begin{abstract}
Many synthesis approaches for reversible and quantum logic have been proposed so far. However, most of them generate circuits with respect to simple metrics, i.e.~gate count or quantum cost. On the other hand, to physically realize reversible and quantum hardware, additional constraints exist. In this paper, we describe cost metrics beyond gate count and quantum cost that should be considered while synthesizing reversible and quantum logic for the respective target technologies. We show that the evaluation of a synthesis approach may differ if additional costs are applied. In addition, a new cost metric, namely Nearest Neighbor Cost (NNC) which is imposed by realistic physical quantum architectures, is considered in detail. We discuss how existing synthesis flows can be extended to generate optimal circuits with respect to NNC while still keeping the quantum cost small.
\end{abstract}

\newcounter {TCounter}
\newcounter {LCounter}
\newcounter {PCounter}
\newcounter {CCounter}
\newcounter {DCounter}
\newcounter {ECounter}

\newtheorem{theorem}[TCounter]{\textbf{Theorem}}
\newtheorem{lemma}[LCounter]{\textbf{Lemma}}
\newtheorem{proposition}[PCounter]{\textbf{Proposition}}
\newtheorem{corollary}[CCounter]{\textbf{Corollary}}
\newtheorem{definition}[DCounter]{\textbf{Definition}}
\newtheorem{example}[ECounter]{\textbf{Example}}
\newproof{pf}{Proof}

\newenvironment{remark}[1][\textbf{Remark}]{\begin{trivlist}
\item[\hskip \labelsep {\bfseries #1}]}{\end{trivlist}}

\newenvironment{changemargin}[2]{%
  \begin{list}{}{%
    \setlength{\topsep}{0pt}%
    \setlength{\leftmargin}{#1}%
    \setlength{\rightmargin}{#2}%
    \setlength{\listparindent}{\parindent}%
    \setlength{\itemindent}{\parindent}%
    \setlength{\parsep}{\parskip}%
  }%
  \item[]}{\end{list}}

\section{Introduction}

Power dissipation becomes an important issue for designing high performance digital circuits. While a significant part of energy dissipation is due to the non-ideal behavior of transistors and materials, the other part is caused by the fundamental Landauer principle introduced in 1961  \cite{Landauer61}. Landauer proved that using conventional (irreversible) logic gates leads to a certain energy dissipation regardless of the underlying technology. More precisely, energy is dissipated each time a bit is lost during the computation. 
In 1973, Bennett showed that to avoid power dissipation in a circuit, it must be built from reversible gates \cite{Bennett73}.
Currently, the field of reversible computing has received considerable attention in particular in low-power CMOS design~\cite{Schrom98, DD:2002}.

Furthermore, quantum gates are inherently reversible so that reversible logic builds the basis in the domain of quantum computation. 
It has been shown that quantum computing significantly improves the rate of advance in processing power for dedicated applications~\cite {Nielsen00}. 
For example the factorization problem cannot be executed on a classical Turing machine as efficiently as on a quantum computer. 
Since conventional CMOS technology suffers from the 
miniaturization and the exponential growth of the number of transistors, 
quantum computing received significant attention as a promising alternative in the last years.


However, quantum algorithms and reversible logic will finally need hardware for the implementation. So far, several technologies have been developed where each one has its own strengths and drawbacks (see e.g.~\cite{DD:2002,MaslovarXiv,Ross08}). In particular, current quantum technologies lead to completely new challenges for the researchers which need to be resolved for having a practical device. Thus, to reach a realizable hardware, synthesis approaches have to consider the physical limitations as well. Nevertheless, most of the currently available synthesis algorithms often only consider gate count and quantum cost 
(see e.g.~\cite {MDM:2005,MaslovTODAES07, SaeediICCAD07, Gupta06, Wille07, GrosseMVL08,WilleDAC09}). 
But depending on the application, the addressed technology, or the final quantum architecture other cost metrics beyond quantum cost and gate count have to be addressed as well.

In this paper, we first give 
an overview of technology-related metrics that should additionally be considered while synthesizing quantum or reversible logic. Thereby, we focus on metrics which can be modelled in such a way that they are still applicable to existing synthesis approaches. In doing so, the applicability of the available synthesis methods is ensured while the results can be optimized
with respect to constraints closer to the physical realization. 
Our analyses reveal that considering these metrics may lead to new conclusions about the performance of synthesis approaches in a respective domain. 

Second, we demonstrate how existing synthesis flows can be extended to optimize the resulting circuits with respect to one of these new cost criteria. To this end, we propose methods to optimize the \emph{Nearest Neighbor Cost} (NNC) for \emph{Linear Nearest Neighbor} (LNN) architectures where only adjacent qubits may interact with each other. LNN architectures are often considered as an appropriate approximation of a scalable quantum architecture (see e.g.~\cite{FDH:2004,Kut:2006,MaslovarXiv}). While NNC optimality can be achieved 
by applying additional SWAP gates, it increases the quantum cost by about one order of magnitude. Thus, we propose improvements
that reduce the resulting quantum cost by more than 50\% on average (83\% in the best case).
As a result, by considering a new technological constraint, a synthesis flow is suggested that goes beyond the previous synthesis paradigms.

The rest of this paper is organized as follows. In Section \ref{sec:prelim} basic concepts are introduced. In Section \ref{sec:cost} a comprehensive overview of technology-related metrics for the synthesis stage is given and discussed. The NNC-optimal synthesis flow for quantum logic is proposed in Section \ref{sec:nnc_synth} followed by the optimization methods presented in Section \ref{sec:opt}. Our experimental results are described in Section \ref{sec:exp}. Finally, we 
draw some conclusions and sketch future work in Section~\ref{sec:conclusion}. 

\section {Background} \label {sec:prelim}

\subsection{Reversible Logic}

A reversible function $f:\mathbb{B}^n\rightarrow \mathbb{B}^n$ over variables $X=\{x_1,\dots , x_n\}$
maps each input assignment to a unique output assignment. Such function must have the same number of input and output variables. A circuit realizing a reversible function is a cascade of reversible gates.
Common reversible gates include:
\begin{itemize}
\item A {\em multiple control Toffoli gate} $t_m$ \cite{Toffoli80} has the form $t_m(C, t)$, where $C = \{x_{i_1} , \dots , x_{i_m}\} \subset X$ is the set of control lines and $t = \{x_j\}$ with $C \cap t = \emptyset$ is the target line. 
The value of the target line is inverted iff all control lines are assigned to~1. For $m$=0 and $m$=1, the gates are called \textit{NOT} and \textit{CNOT}, respectively. For $m$=2, the gate is called \textit{C$^2$NOT} or \textit{Toffoli}.
\item A {\em multiple control Fredkin gate} $f_m$ \cite{Fredkin82} has two target lines and $m$ control lines. The gate interchanges the values of the target lines iff the conjunction of all $m$ control lines evaluates to 1. For $m$=0, the gate is called \textit{SWAP} gate.
\item A {\em Peres gate} $P$~\cite{Peres85} has one control line $x_i$ and two target lines $x_{j_1}$ and $x_{j_2}$, and 
is a $t_2(\{x_{i},x_{j_1}\}, x_{j_2})$ and a $t_1(\{x_{i}\}, x_{j_1})$ in a cascade.
\end{itemize}

In this paper, $n$ is particularly used to denote the number of inputs and outputs. 
Outputs (inputs) that are not required in the function specification are considered as \textit{garbage} (\textit{constant}) lines. 
The notations $n_{c}$ and $n_{g}$ are used as the number of constant inputs and number of garbage outputs, respectively.

\subsection{Decomposition to Quantum Logic}\label{sec:prelim_decomp}

Quantum logic is inherently reversible \cite{Nielsen00} and manipulates qubits rather than pure logic values. A \textit{qubit} is a two-level quantum system, described by a two-dimensional complex Hilbert space. The two orthogonal quantum states are used to represent the values 0 and~1. The state of a qubit for two pure logic states can be expressed as \mbox{$ | \Psi \rangle = \alpha |0 \rangle + \beta |1 \rangle$} (called \textit{superposition}), where
$\alpha$ and $\beta$ are complex numbers so that $|\alpha|^2 + |\beta|^2 = 1$.

Each Toffoli, Fredkin, or Peres gate can be decomposed into a quantum circuit composed of a sequence of \emph{elementary quantum gates}~\cite{Nielsen00} defined as follows: 
\begin{itemize}
\item Inverter (NOT): A single qubit is inverted.
\item Controlled-{NOT} (CNOT): The target qubit is inverted if the control qubit is 1.
\item Controlled-V: Performs the V operation known as the square root of NOT, since two consecutive V operations are equivalent to an inversion.
\item Controlled-V$^+$: Performs the inverse of V .
\end{itemize}

Fig.~\ref{fig:example_rev} shows a Toffoli and a Fredkin gate in a cascade. The resulting (decomposed) quantum circuit is depicted in Fig.~\ref{fig:example_decomp}.

\begin{figure}
\centering
\subfigure[Circuit\label{fig:example_rev}]{
\centering
\includegraphics[scale=.23]{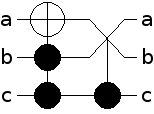}
}~~\subfigure[Decomposition (quantum circuit)\label{fig:example_decomp}]{
\centering
\includegraphics[scale=.23]{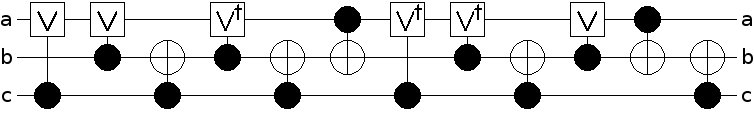}
}
\caption{Reversible circuit and its decomposed quantum circuit}
\end{figure}

\section {Cost Metrics} \label {sec:cost}
In order to address various technical constraints in the synthesis stage, metrics have been proposed that are used to evaluate the respective results. Among them, gate count and quantum cost received most attractions and are widely applied (e.g. see \cite {MDM:2005,MaslovTODAES07, SaeediICCAD07, Gupta06, Wille07, GrosseMVL08,WilleDAC09}). 
However, other cost metrics imposed by target technologies exist and need to be considered.

The aim of this section is to introduce those metrics which can be used to guide a synthesis tool to construct realizable circuits. It is worth noting that there are other (physical) metrics like duration of pulse sequences and number of traps that their considerations need some interactions between various steps of the design cycle of quantum and reversible circuits. However, the focus of this paper is on those cost metrics that can be used in the synthesis step.

While the introduced metrics are related to the physical realizations of quantum and reversible logic more than the ones have been applied so far, they can be simply used by the available synthesis tools too. As the experimental results show such considerations can lead to a different (and more realistic) evaluation of synthesis approaches.

\subsection{Number of Lines and Constant Inputs}

Number of lines $n$, constant inputs $n_{c}$, and garbage outputs $n_{g}$ are common metrics to measure the quality of synthesized circuits. Even if these measures are negligible for reversible CMOS technologies (see e.g.~\cite{DD:2002}) they have high importance in quantum computation
where each qubit must be represented by a physical entity
that supports
distinct and superposition states.
Currently, such entities are not arbitrarily available in quantum technologies\footnote{Current systems contain e.g.~12~\cite{Scientist} or 28 qubits~\cite{DWave}, respectively.}. Besides that, initializing quantum registers
cannot be simply done because of the exponential state-space of an $n$-qubit register (see the method proposed in \cite{Shende06}). Thus, synthesis approaches should keep track of $n_{c}$, $n_{g}$, and $n$ if quantum circuits are particularly addressed.

\subsection{Gate Count and Quantum Cost}

Number of gates has been used to evaluate nearly all synthesis approaches so far. For an arbitrary circuit $C$ with $k$ gates $g_{1}$, $g_{2}$, $\cdots$, $g_{k}$, the \textit{gate count} metric is denoted as $gc$ and defined as $gc=k$. Besides that, \emph
{quantum cost} $qc$ are used to measure the implementation cost of quantum circuits. 
More precisely, quantum cost is defined as the number of elementary quantum operations needed to realize a gate (see Table~\ref{QuantumCost}).
In addition, the quantum cost for a circuit is defined as $qc=\sum_{i=1}^{k} qc_{g_{i}}$.
Obviously, quantum cost should be considered in particular if quantum circuits are addressed.
\begin{table}[!t]
\caption{Quantum costs of various gates}
\label{QuantumCost}
\centering
\begin{tabular}{|c|c|}
\hline
Gate Type & Quantum Cost \\
\hline
\hline
$t_0,t_1$ (NOT, CNOT) & 1 \cite{Nielsen00}\\
\hline
$t_2$ (Toffoli) & 5 \cite{Barenco95}\\
\hline
$t_{m}$ $(3 \leq m \leq \left\lceil\frac{n}{2}\right\rceil)$ & $12 \times m - 22$ \cite{MaslovDATE05} \\ 
\hline
$t_{m}$ $(\left\lceil\frac{n}{2}\right\rceil+1 \leq m \leq n-2 )$ & $24\times m - 64$ \cite{MaslovDATE05}\\
\hline
$t_{n-1}$ & $ 2^{n} - 3$ \cite{Barenco95}\\
\hline
$f_0$ (SWAP) & 3 \cite{Nielsen00}\\
\hline
$f_m$ $(1 \leq m \leq n-2)$ & $2+Cost_{t_{m+1}}$ \cite{Nielsen00}\\
\hline
Peres & 4 \cite{Peres85}\\
\hline
\end{tabular}
\end{table}


\subsection{Transistor Cost}
In order to realize Toffoli or Fredkin gates on reversible CMOS technologies (as done e.g.~in \cite{DD:2002}), a number of transistors is required where its complexity depends on the number of control lines of the respective gate. This metric is denoted as $TrC$ in the rest of the paper. It is defined by $8\cdot m$ where $m$ is the number of control lines of a given gate \cite{ar:TG08}. The $TrC$ of a circuit is defined as the sum of the $TrC$s of its gates.




\subsection{Nearest Neighbor Cost (NNC)}

Although most quantum algorithms presume that interaction between arbitrary qubits is possible with no extra cost, some restrictions exist in real quantum technologies. As an example, in a \textit{Linear Nearest Neighbor} (LNN) architecture only adjacent qubits are allowed to interact\footnote{The LNN architecture is often considered as an appropriate approximation to a scalable quantum architecture. If one can show that a circuit can be efficiently reorganized to be executed in the LNN architecture, such a circuit could be run efficiently in many other architectures~\cite{FDH:2004,Kut:2006,MaslovarXiv}.}. Hence, gates of the same type (e.g.~all $\left(^{n}_{2}\right)\cdot(n-2)$ 3-bit Toffoli gates) do not necessarily have the same implementation cost.
To measure this, a new cost metric denoted \emph{NNC} is introduced:

Consider a 2-qubit quantum gate $g$ where its control and target are placed at the $c^{th}$ and $t^{th}$ line ($0\leq c,t <n$), respectively. The NNC of $g$ is defined as $|c-t-1|$, i.e. distance between control and target lines. The NNC of a circuit is defined as the sum of the NNCs of its gates. Optimal NNC for a circuit is 0 where all quantum gates are either 1-qubit or 2-qubit gates performed on adjacent qubits.

{\setlength{\tabcolsep}{3pt}
\begin{table*}[t]
\caption{Comparison of Different Cost Metrics}\label{table:metrics}
\scriptsize
\begin{changemargin}{-3cm}{-3cm}     
\begin{tabular}{|r|rr|rr|rr|rr|rr|rr|rr|rr|rr|}

\hline
  Benchmark         &\multicolumn{ 2}{|c|}{$n$}&\multicolumn{ 2}{|c|}{$n_g$}  &   \multicolumn{ 2}{|c|}{$n_c$} & \multicolumn{ 2}{|c|}{$gc$} & \multicolumn{ 2}{|c|}{$qc$}& \multicolumn{ 2}{|c|}{NNC}& \multicolumn{ 2}{|c|}{$Depth$}&  \multicolumn{ 2}{|c|}{$Dis_{avg}$} &   \multicolumn{ 2}{|c|}{$TrC$} \\

           &        RMS &        BDD &        RMS &        BDD &        RMS &        BDD &        RMS &        BDD &        RMS &        BDD &        RMS &        BDD &        RMS &        BDD &        RMS &        BDD &        RMS &        BDD  \\
\hline

 4mod5\_8   &    {\bf 5} &          7 &    {\bf 4} &          6 &    {\bf 1} &          3 &          8 &          8 &   {\bf 12} &         24 &   {\bf 16} &         41 &   {\bf 10} &         24 &  {\bf 4,6} &       6,71 &  {\bf 64} &        88 \\

   alu\_9   &    {\bf 5} &          7 &    {\bf 4} &          6 &    {\bf 0} &          2 &         13 &    {\bf 9} &         45 &   {\bf 29} &   {\bf 39} &         45 &         42 &   {\bf 28} &       18,2 & {\bf 8,14} &        152 &  {\bf 104} \\

 decod24\_10   &    {\bf 4} &          6 &    {\bf 0} &          2 &    {\bf 2} &          4 &         18 &   {\bf 11} &         86 &   {\bf 27} &         88 &   {\bf 33} &         82 &   {\bf 26} &       35,8 &  {\bf 8,5} &        240 &  {\bf 96} \\

 hwb9\_65   &   {\bf 9} &        170 &    {\bf 0} &        161 &    {\bf 0} &        161 &       2223 &  {\bf 699} &      23178 & {\bf 2275} & {\bf 43624} &     118639 &      18022 & {\bf 1997} &     4194,6 & {\bf 26,62} &      31136 & {\bf 8544} \\

%
%
%
%
\hline

\end{tabular}
\end{changemargin}
\end{table*}
} 

\subsection{Circuit Depth}

Consider two consecutive gates $g_{i}$ and $g_j$ with control sets $C_{i}$ and $C_{j}$ and with target sets $T_{i}$ and $T_{j}$, respectively. These gates can \textit{concurrently} be applied if $C_{i} \cap T_{j}=\emptyset$ and $C_{j} \cap T_{i}=\emptyset$ ($C_{i}$ and $C_{j}$ may have some common elements). Suppose a quantum circuit $C$ with $k$ elementary gates. Assume that $C$ contains $m$ subcircuits with concurrent gates where the $i^{th} (1\leq i \leq m)$ subcircuit contains $k_i$ concurrent elementary gates. The circuit depth $Depth$ is defined as the number of steps required to execute all available gates in a circuit, i.e. $Depth=k - \sum\limits_{i = 1}^m {(k_i - 1)}$.

Since the time which a qubit can keep its quantum state (the \textit{coherence time}) and the time needed to perform a gate (the \textit{gate operation time}) may vary from one technology to another (for example see Table III of \cite{Meter06}), considering the circuit depth at the synthesis stage for quantum circuits is vital. Despite the fact that quantum algorithms already exploit algorithmic parallelism to increase the processing speed, synthesis algorithms should produce concurrent gates for efficient quantum circuit implementation.

\subsection{Gate Distribution}
As mentioned above, coherence time for qubits and operation time for gates are widely affected by technological parameters.
It can be verified that the total operation time of gates applied to a qubit must be less than its qubit de-coherence time; otherwise the qubit value is lost before applying all gates. The number of elementary gates applied to the $i^{th}$ qubit of a quantum circuit is denoted as $Dis_{i}$. Average, minimum, and maximum number of applied elementary gates are denoted as $Dis_{avg}$, $Dis_{min}$, and $Dis_{max}$, respectively. By considering $Dis_{min}$ and $Dis_{max}$, designers may want to balance the distributions of all qubits for quantum circuits.


\subsection {Discussion}

As an example of the various cost metrics,
Fig.~\ref {sampleCircuit} shows a circuit with $n=5$, $n_c=1$, $n_g=4$, $gc=7$, $qc=7$, NNC = 9, $Depth=4$, $Dis_{avg}=2.6$, $Dis_{min}=2$, $Dis_{max}=3$, and $TrC = 48$.
Although, gate count ($gc$) and quantum cost ($qc$) have been extensively used for the evaluation of circuit quality so far, consideration of additional metrics as introduced above may lead to different judgements about the quality of a circuit.
To illustrate this, Table \ref{table:metrics} compares different costs of some benchmark circuits obtained from two different synthesis approaches, namely the Reed-Muller spectra approach~\cite{MaslovTODAES07} (denoted by \textsc{RMS}) and the BDD-based method~\cite{WilleDAC09} (denoted by \textsc{BDD})\footnote{We like to thank the authors of both approaches for making us their tool available.}.

\begin{figure}[tb]
	\centering
		\includegraphics[scale=0.7]{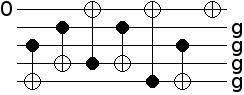}
\caption{Sample circuit}
	\label{sampleCircuit}
\end{figure}

As shown in Table \ref{table:metrics}, if only gate count and quantum cost are used for the evaluation, one must conclude that the BDD-based approach leads to the best results. However, if further cost metrics are considered the results become very different. For example, the RMS approach leads to better realizations in terms of NNC and number of lines.
That is, if quantum circuits should be built, the results obtained by this approach are better.
In contrast, if reversible CMOS circuits are addressed, the number of lines are negligible and NNC can be ignored. Since additionally the transistor count is lower, the circuits obtained by the BDD-based method are the better choice for this domain.

Altogether, cost metrics beyond gate count and quantum cost should be used to control the synthesis so that circuits with special properties (e.g.~few circuit lines, an NNC of 0, etc.) result. 
But the consideration of various constraints at a single step complicates the respective algorithms. Nevertheless, it may be acceptable to use a multi-stage design flow that successively address all required constraints. Thereby, approaching important metrics at earlier stages of a design flow is preferred.
In the rest of the paper, we show how an existing design flow can be adjusted so that a further cost criterion (namely NNC) is supported.

\section {NNC-optimal Synthesis of Quantum Logic} \label {sec:nnc_synth}

Quantum circuits can be synthesized using multiple control Toffoli gates first that are afterwards mapped to elementary quantum gates. On the other hand, elementary quantum gates can be directly applied during the synthesis process. While for the latter case, only small circuits have been determined so far (e.g.~see \cite{Shende06, Hung06,GrosseMVL08}), approaches for Toffoli network synthesis can handle larger functions and circuits
(e.g.~see~\cite {MDM:2005,MaslovTODAES07, SaeediICCAD07, Gupta06, Wille07, WilleDAC09}). 
However, in both cases often only the gate count or the quantum cost are used as cost metric. In this section, we show how these synthesis flows can be extended to additionally address NNC.

Current decomposition algorithms may lead to non-optimal circuits with respect to NNC. As an example, Fig. \ref{fig:decomp}~(a) shows the standard decomposition of a Toffoli gate leading to an NNC value of 1.
The general idea of our NNC optimization is to apply adjacent SWAP gates whenever a non-adjacent quantum gate occurs in the standard decomposition. More precisely, SWAP gates are added in front of each gate $g$ with non-adjacent control and target lines to ``move'' a control (target) line of $g$ towards the target (control) line until they become adjacent. Afterwards, SWAP gates are added to restore the original ordering of circuit lines. In total, this leads to additional quantum cost given by the following lemma:

\begin{lemma} \label{NNTheorem}
Consider a quantum gate $g$ where its control and target are placed at the $c^{th}$ and $t^{th}$ lines, respectively. Using adjacent SWAP gates as proposed, additional quantum cost of $6 \cdot |c-t-1|$ are needed.
\end{lemma}
\begin{proof}
In total, $|c-t-1|$ adjacent SWAP operations are required to move the control line to the target, so that both become adjacent. Another $|c-t-1|$ SWAP operations are needed to restore the original ordering. Considering quantum cost of 3 for each SWAP operation, this leads to the additional quantum cost of $6 \cdot |c-t-1|$.
\end{proof}

By applying this method consecutively to each non-adjacent gate, a quantum circuit with NNC of 0 can be determined in linear time.


\begin{example}
Consider the standard decomposition of a Toffoli gate as depicted in Fig.~\ref{fig:decomp}~(a). As can be seen, the first gate is non-adjacent. Thus, to achieve NNC-optimality, SWAP gates in front and after the first gate are inserted (see Fig.~\ref{fig:decomp}~(b)). Since each SWAP gate requires 3 quantum gates, this increases the total quantum cost to~$11$ but leads to an NNC value of 0.
\end{example}

In the rest of this paper, this method is denoted by \emph{naive NNC-based decomposition}. Obviously, this straightforward method leads to a high increase in quantum cost. In the next section, more elaborated approaches for synthesizing NNC-optimal circuits are proposed.


\begin{figure}[tb]

\begin{picture}(240,160)

\put(0,90){\includegraphics[scale=.0055]{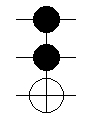}}

\put(35,120){\vector(2,1){40}}
\put(35,135){\footnotesize Standard}
\put(35,100){\vector(2,-1){40}}
\put(39,100){\footnotesize NNC-based}
\put(55,92){\footnotesize (naive)}
\put(20,85){\vector(1,-3){14}}
\put(31,60){\footnotesize NNC-based}
\put(40,53){\footnotesize (exact)}

\put(80,125){\includegraphics[scale=.3]{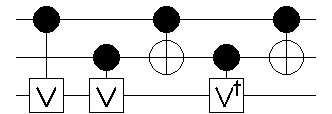}}
\put(128,119){\footnotesize (a)}
\put(178,145){\footnotesize Quantum cost: 5}
\put(178,135){\footnotesize NNC: 1}

\put(80,90){\includegraphics[scale=.3]{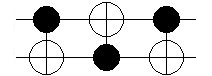}}
\put(128,38){\footnotesize (b)}
\put(82,88){\dashbox{2}(61,26)}
\put(80,45){\includegraphics[scale=.3]{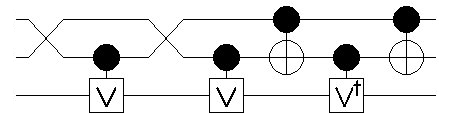}}
\put(94,76){\line(0,1){12}}
\put(86,60){\dashbox{2}(16,16)}
\put(130,76){\line(0,1){12}}
\put(122,60){\dashbox{2}(16,16)}
\put(178,100){\footnotesize Quantum cost: 11}
\put(178,90){\footnotesize NNC: 0}

\put(5,2){\includegraphics[scale=.3]{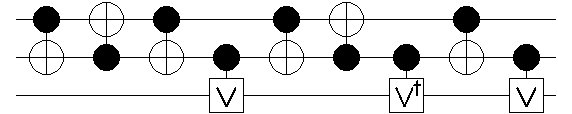}}
\put(80,-5){\footnotesize (c)}
\put(178,22){\footnotesize Quantum cost: 9}
\put(178,12){\footnotesize NNC: 0}

\end{picture}
\caption{Different decompositions of a Toffoli gate}
\label{fig:decomp}
\end{figure}
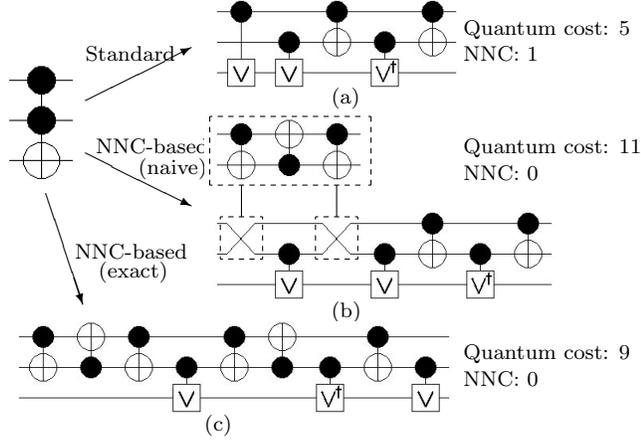

\section {Improvements} \label {sec:opt}

Two improved approaches for NNC-optimal generation of quantum circuits from reversible logic are introduced. The first one exploits exact synthesis techniques while the second one manipulates the circuit and specification, respectively.


\begin{figure*}[t]
\centering
\includegraphics[scale=0.33]{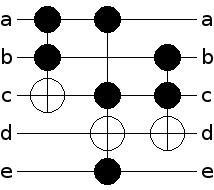}
\vspace{.4cm}
\caption{Circuit of Example \ref {exmacro} }\label{fig:example_circuit}
\end{figure*}
\begin{figure*}[t]
\centering
\subfigure[\label{fig:motivation_reordering1}]{
\includegraphics[scale=.23]{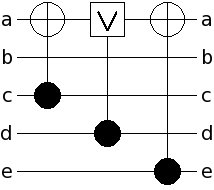}
}~~\subfigure[\label{fig:motivation_reordering2}]{
\includegraphics[scale=.23]{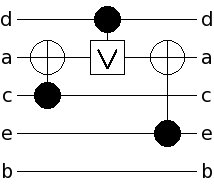}
}
\caption{Reordering circuit lines}
\end{figure*}
\begin{figure*}[t]
\centering
\subfigure[Original circuit\label{fig:reordering1}]{
\includegraphics[scale=.3]{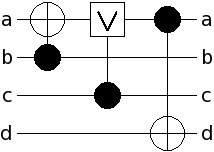}
}\subfigure[Global reordering\label{fig:reordering_global2}]{
\includegraphics[scale=0.37]{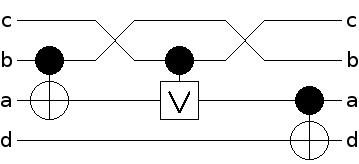}
}
\subfigure[Local reordering\label{fig:reordering_local2}]{
\includegraphics[scale=0.3]{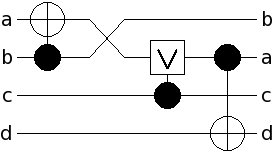}
}
\caption{Global and local reordering}
\end{figure*}

\subsection{Exploiting Exact Synthesis}\label {sec:opt_exact}

A few exact synthesis methods for quantum and reversible circuits have been recently introduced that generate quantum circuits with minimal quantum cost (for examples see \cite {Hung06,GrosseMVL08}). In this section, we propose a synthesis algorithm to construct quantum circuits with \emph{both} minimal quantum cost and minimal NNC.


The developed approach is similar to the one introduced in \cite{GrosseMVL08}.
Here, the synthesis problem is expressed as a sequence of \emph{Boolean satisfiability} (SAT) instances.
For a given function $f$, it is checked if a circuit with $c$ gates realizing $f$ exists. Thereby, $c$ is initially assigned to $1$ and increased in each iteration if no realization is found.



More formally, for a given $c$ and a reversible function $f:\mathbb{B}^n \rightarrow \mathbb{B}^n$, the following SAT instance is created:

$$\Phi \wedge \bigwedge\limits_{i=0}^{2^n-1}([\overrightarrow{inp}_i]_2 = i \wedge [\overrightarrow{out}_i]_2 = f(i))\mbox{,}$$
where
\begin{itemize}
\item $\overrightarrow{inp}_i$ is a Boolean vector representing the inputs of the network to be synthesized for
truth table line~$i$,
\item $\overrightarrow{out}_i$ is a Boolean vector representing the outputs of the network to be synthesized for
truth table line~$i$, and
\item $\Phi$ is a set of constraints representing the synthesis problem for a given gate library.
\end{itemize}
The difference in comparison to~\cite{GrosseMVL08} is, that the constraints in $\Phi$ do not represent the whole set of elementary quantum gates. In fact, a restricted gate library with only adjacent gates is applied.


Although solving the generated SAT instances using a modern SAT solver can produce optimized circuits, the applicability of such an exact method is always limited to relatively small functions due to the exponential search space. In this sense, the proposed method is sufficient to construct
minimal realizations with respect to both quantum cost and NNC for a set of Toffoli and Peres gate configurations as shown in Table~\ref {tab:MacrosTable}.

But nevertheless, these results can be exploited to improve the naive NNC-based decomposition.
Once an exact NNC-optimal quantum circuit for a reversible gate is available (denoted by \emph{macro} in the following), it can be reused as shown by the following example:

\begin{example} \label {exmacro}
Reconsider the decomposition of a Toffoli gate as depicted in Fig. \ref{fig:decomp}. Using the proposed exact synthesis approach, a minimal quantum circuit (with respect to both quantum cost and NNC) as shown in Fig.~\ref{fig:decomp}(c) is determined\footnote{The circuit is minimal with respect to the underlying gate library introduced in Section~\ref{sec:prelim_decomp}. If another library is applied (e.g.~\cite{SK:2003}), better realizations may be possible.}. In comparison to the
naive method (see Fig.~\ref{fig:decomp}(b)), this reduces the quantum cost from 11 to 9 while still ensuring NNC optimality. Furthermore, the realization can be reused as a macro while decomposing larger reversible circuits. For example, consider the circuit shown in Fig.~\ref{fig:example_circuit}. Here, for the second gate the naive method is applied (i.e.~standard decomposition is performed and SWAPs are added), while for the remaining ones the obtained macro is used. This enables a quantum cost reduction from 96 to 92.
\end{example}

In total, we generated 13 macros as listed in Table~\ref {tab:MacrosTable} together with the respective costs in comparison to the costs obtained by using the naive method. As can be seen, exploiting these macros reduces the cost for each gate by up to 63\%. Besides that, also further macros (e.g.~known from literature) can be additionally applied. The effect of these macros on the decomposition of reversible circuits is considered in our experiments in Section~\ref{sec:exp} in detail.

\begin{table}[!t]
	\caption{List of available macros}
	\label{tab:MacrosTable}
	\centering
		\begin{tabular}{|c|c|c|c|c|}
			\hline
			$n$ & Macro & \multicolumn{ 2}{|c|}{Cost} & Impr\\
			    &       &  Naive &Exact & \\
			\hline
			3 & P(\{a,b\},c), P(\{c,b\},a) & 12 & 8 & 33\%\\
			3 & P(\{a,c\},b), P(\{c,a\},b) & 24& 12  & 50\%\\
			4 & P(\{a,b\},d), P(\{d,c\},a) & 30 & 11 & 63\%\\
			3 & t2(\{a,b\},c), t2(\{c,b\},a) & 11 & 9 & 18\%\\
			4 & t2(\{a,b\},d), t2(\{d,c\},a) & 29 & 12 & 59\%\\
			3 & t2(\{a,c\},b) & 17 & 13 & 24\%\\
			4 & t2(\{d,b\},a), t2(\{a,c\},d) & 29 &13 & 55\%\\
			\hline
		\end{tabular}
\end{table}

\subsection{Reordering Circuit Lines}\label {sec:opt_reordering}

Applying the approaches introduced so far leads to an increase in the quantum cost for each non-adjacent gate. 
In contrast, by modifying the ordering of the circuit lines, some of the additional costs can be saved. As an example, consider the circuit in Fig.~\ref {fig:motivation_reordering1} with quantum cost 3 and an NNC value of 6. By reordering the lines as shown in Fig.~\ref{fig:motivation_reordering2}, the NNC value can be reduced to 1 without increasing the total quantum cost.
It is worth noting that manipulating the line order has been previously done to reduce the quantum cost e.g.~in~\cite{MDM:2005,WGDD:2009}.
To determine which lines should be reordered, two heuristic methods are proposed in the following. The former one changes the ordering of the primary inputs and outputs according to a global view while the latter one applies a local view to assign the line ordering.



\subsubsection{Global Reordering}
After applying the standard decomposition introduced in Section~\ref{sec:prelim_decomp}, a cascade of 1- and 2-qubit gates is generated. Now, an ordering of the circuit lines which reduces the total NNC value is desired. To do that, the ``contribution'' of each line to the total NNC value is calculated. More precisely, for each gate $g$ with control line $i$ and target line $j$, the NNC value is determined. This value is added to variables $imp_i$ and $imp_j$ which are used to save the impacts of the circuit lines $i$ and $j$ on the total NNC value, respectively. Next, the line with the highest NNC impact is chosen for reordering and
placed at the middle line (i.e.~swapped with the middle line). If the selected line is the middle line itself, a line with the next highest impact is selected. This procedure is repeated until no better NNC value is achieved. Finally, SWAP operations as described in the previous sections are added for each 
non-adjacent gate. 

\begin{example}\label{example:global_reordering}
Consider the circuit depicted in Fig. \ref{fig:reordering1}. After calculating the NNC contributions, we have $imp_a=1.5$, $imp_b=0$, $imp_c=0.5$, and $imp_d=1$. Thus, lines $a$ (highest impact) and $c$ (middle line) are swapped. Since further swapping does not improve the NNC value, reordering terminates and SWAP gates are added for the remaining non-adjacent gates. The resulting circuit is depicted in Fig.~\ref {fig:reordering_global2} and has quantum cost of 9 in comparison to 21 that results if the naive method is applied.
\end{example}

\subsubsection{Local Reordering}
In order to save SWAP gates, line ordering can also be applied according to a local schema as follows. The circuit is traversed from the inputs to the outputs. As soon as there is a gate $g$ with an NNC value greater than~0, a SWAP operation is added in front of $g$ to enable an adjacent gate. However, in contrast to the naive NNC-based decomposition, no SWAP operation is added after $g$. Instead, the resulting ordering is used for the rest of the circuit (i.e.~propagated through the remaining circuit). This process is repeated until all gates are traversed. 

\begin{example} \label {exm:localordering}
Reconsider the circuit depicted in Fig.~\ref{fig:reordering1}. 
The first gate is not modified since it has an NNC of~0. For the second gate, a SWAP operation is applied to make it adjacent. Afterwards, the new line ordering is propagated to all remaining gates resulting in the circuit shown in Fig.~\ref{fig:reordering_local2}. This procedure is repeated until the whole circuit has been traversed. 
Finally, again a circuit with quantum cost of 9 (in contrast to 21) results.
\end{example}

\section {Experimental Results} \label {sec:exp}

In this section, experimental results are presented.
We evaluated the methods introduced in Section~\ref{sec:nnc_synth} and Section~\ref{sec:opt}, respectively,
by measuring the overhead needed to synthesize circuits with an optimal NNC value of~0. The proposed approaches have been implemented in C++ and applied to the benchmark circuits available at RevLib \cite{WilleMVL08}.
All experiments have been carried out on an AMD Athlon 3500+ with 1~GB of main memory.





The results are shown in Table~\ref{table:nnc}. The first column gives the names of the circuits followed by unique identifiers as used in RevLib. Then, the number of circuit lines (\textsc{n}), the gate count (\textsc{gc}), the quantum cost (\textsc{qc}), and the NNC value of the original (reversible) circuits are shown. The following columns denote the quantum cost of the NNC-optimal circuits obtained by using the naive method 
(\textsc{Naive}), by additionally exploiting macros
(\textsc{+Macros}), and by applying reordering as described in Section~\ref{sec:opt_reordering} (\textsc{Global}, \textsc{Local}, or both), respectively. The next column gives the percentage of the best quantum cost reduction obtained by the improvements in comparison to the naive method (\textsc{Best Imprmnts}). The last column shows 
the smallest overhead in terms of quantum cost needed to achieve NNC-optimality in comparison to the original circuit (\textsc{Overhead to achieve NNC optimality}). All run-times are negligible (i.e.~less than one minute) 
and are omitted in the table.

{\setlength{\tabcolsep}{3pt}
\begin{table*}[t]
\caption{NNC-based synthesis}\label{table:nnc}
\scriptsize
\begin{changemargin}{-2cm}{-2cm}
\begin{tabular}{|r||rrrr||rr|rrr|r||r||}
\hline
\textsc{Benchmark}            & \multicolumn{ 4}{|c||}{\textsc{Original Circuit}} &                                    \multicolumn{ 6}{|c||}{\textsc{Decomposed (NNC-optimal) Circuits}} &            \textsc{Overhead to} \\

           &            &            &            &            &  & & \multicolumn{ 3}{|c|}{\textsc{Reordering}} &       \textsc{Best} & \textsc{achieve NNC} \\

           &            &            &            &            &                 \textsc{Naive} & \textsc{+Macros} &     \textsc{Global} &      \textsc{Local} & \textsc{Glob.+Loc.} &   \textsc{Imprmnts} & \textsc{optimality} \\

           &          \textsc{n} &       \textsc{gc}   &       \textsc{qc}   &      \textsc{NNC}   &        \textsc{qc}   &       \textsc{qc}   &       \textsc{qc}   &       \textsc{qc}   &        \textsc{qc}  &            &            \\
\hline
\hline
 0410184\_169   &         14 &         46 &         90 &         24 &        234 &  {\bf 197} &        234 &        423 &        423 &       16\% &       2,19 \\
\hline
 3\_17\_13   &          3 &          6 &         14 &          3 &         32 &   {\bf 28} &         32 &         32 &         32 &       13\% &       2,00 \\
\hline
 4\_49\_17   &          4 &         12 &         32 &         21 &        158 &        120 &        128 &   {\bf 98} &   {\bf 98} &       38\% &       3,06 \\
\hline
 4gt10-v1\_81   &          5 &          6 &         34 &         41 &        282 &        282 &        258 &        150 &  {\bf 147} &       48\% &       4,32 \\
\hline
 4gt11\_84   &          5 &          3 &          7 &          7 &         49 &         47 &         25 &         22 &   {\bf 16} &       67\% &       2,29 \\
\hline
 4gt12-v1\_89   &          5 &          5 &         42 &         80 &        525 &        525 &        321 &        171 &  {\bf 168} &       68\% &       4,00 \\
\hline
 4gt13-v1\_93   &          5 &          4 &         16 &         26 &        173 &        173 &         77 &         56 &   {\bf 53} &       69\% &       3,31 \\
\hline
 4gt4-v0\_80   &          5 &          5 &         34 &         55 &        366 &        364 &        168 &  {\bf 138} &        141 &       62\% &       4,06 \\
\hline
 4gt5\_75   &          5 &          5 &         21 &         20 &        142 &        138 &        118 &         82 &   {\bf 79} &       44\% &       3,76 \\
\hline
 4mod5-v1\_23   &          5 &          8 &         24 &         25 &        174 &        155 &        114 &   {\bf 78} &   {\bf 78} &       55\% &       3,25 \\
\hline
 4mod7-v0\_95   &          5 &          6 &         38 &         36 &        256 &        256 &        352 &        127 &  {\bf 121} &       53\% &       3,18 \\
\hline
 add16\_174   &         49 &         64 &        192 &         95 &        762 &  {\bf 473} &        762 &       1104 &       1104 &       38\% &       2,46 \\
\hline
 add32\_183   &         97 &        128 &        384 &        191 &       1530 &  {\bf 953} &       1530 &       3744 &       3744 &       38\% &       2,48 \\
\hline
 add64\_184   &        193 &        256 &        768 &        383 &       3066 & {\bf 1913} &       3066 &      13632 &      13632 &       38\% &       2,49 \\
\hline
 add8\_172   &         25 &         32 &         96 &         47 &        378 &  {\bf 233} &        378 &        360 &        360 &       38\% &       2,43 \\
\hline
 aj-e11\_165   &          4 &         13 &         45 &         39 &        280 &        260 &        280 &  {\bf 181} &  {\bf 181} &       35\% &       4,02 \\
\hline
 alu-v4\_36   &          5 &          7 &         31 &         35 &        242 &        238 &        218 &        113 &  {\bf 104} &       57\% &       3,35 \\
\hline
 cnt3-5\_180   &         16 &         20 &        120 &        416 &       2621 &       2591 &       1457 &        731 &  {\bf 728} &       72\% &       6,07 \\
\hline
 cycle10\_2\_110   &         12 &         19 &       1126 &       3368 &      21420 &      21420 &      21420 & {\bf 8046} & {\bf 8046} &       62\% &       7,15 \\
\hline
 decod24-v3\_46   &          4 &          9 &          9 &          9 &         63 &         63 &         39 &   {\bf 21} &         24 &       67\% &       2,33 \\
\hline
 ham15\_108   &         15 &         70 &        453 &       2506 &      15494 &      15390 &      14030 &       2627 & {\bf 2588} &       83\% &       5,71 \\
\hline
 ham7\_104   &          7 &         23 &         83 &        158 &       1035 &       1027 &        657 &        342 &  {\bf 333} &       68\% &       4,01 \\
\hline
 hwb4\_52   &          4 &         11 &         23 &         14 &        107 &         83 &        107 &   {\bf 65} &   {\bf 65} &       39\% &       2,83 \\
\hline
 hwb5\_55   &          5 &         24 &        104 &        119 &        823 &        817 &        595 &  {\bf 337} &        340 &       59\% &       3,24 \\
\hline
 hwb6\_58   &          6 &         42 &        142 &        193 &       1304 &       1160 &       1268 &        614 &  {\bf 545} &       58\% &       3,84 \\
\hline
 hwb7\_62   &          7 &        331 &       2325 &       4236 &      27967 &      27869 &      25939 &      13390 & {\bf 12955} &       54\% &       5,57 \\
\hline
 hwb8\_118   &          8 &        633 &      14260 &      28803 &     187272 &     186880 &     182196 & {\bf 87495} &      87498 &       53\% &       6,14 \\
\hline
 hwb9\_123   &          9 &       1959 &      18124 &      47373 &     304659 &     304540 &     302481 &     124068 & {\bf 124041} &       59\% &       6,84 \\
\hline
 mod5adder\_128   &          6 &         15 &         83 &        154 &       1011 &        978 &        675 &  {\bf 330} &        333 &       67\% &       3,98 \\
\hline
 mod8-10\_177   &          5 &         14 &         88 &        147 &        975 &        969 &        621 &        372 &  {\bf 363} &       63\% &       4,13 \\
\hline
 plus127mod8192\_162   &         13 &        910 &      57400 &     165415 &    1057946 &    1057804 &    1057946 & {\bf 503516} & {\bf 503516} &       52\% &       8,77 \\
\hline
 plus63mod4096\_163   &         12 &        429 &      25492 &      63732 &     407926 &     407784 &     407926 & {\bf 210400} & {\bf 210400} &       48\% &       8,25 \\
\hline
 plus63mod8192\_164   &         13 &        492 &      32578 &      99482 &     633994 &     633852 &     633994 & {\bf 279016} & {\bf 279016} &       56\% &       8,56 \\
\hline
 rd32-v0\_67   &          4 &          2 &         10 &          5 &         38 &   {\bf 19} &         20 &         32 &         20 &       50\% &       1,90 \\
\hline
 rd53\_135   &          7 &         16 &         77 &        124 &        822 &        750 &        702 &        330 &  {\bf 303} &       63\% &       3,94 \\
\hline
 rd73\_140   &         10 &         20 &         76 &        119 &        790 &        739 &        646 &        304 &  {\bf 295} &       63\% &       3,88 \\
\hline
 rd84\_142   &         15 &         28 &        112 &        234 &       1516 &       1465 &       1696 &  {\bf 556} &        586 &       63\% &       4,96 \\
\hline
 sym9\_148   &         10 &        210 &       4368 &      12184 &      77556 &      77556 &      67428 & {\bf 20643} &      25023 &       73\% &       4,73 \\
\hline
 sys6-v0\_144   &         10 &         15 &         67 &         96 &        638 &        587 &        842 &  {\bf 263} &        308 &       59\% &       3,93 \\
\hline
 urf1\_149   &          9 &      11554 &      57770 &     122802 &     794582 &     735170 &     659150 & {\bf 238475} &     238490 &       70\% &       4,13 \\
\hline
 urf2\_152   &          8 &       5030 &      25150 &      45338 &     297178 &     276882 &     297178 & {\bf 101683} & {\bf 101683} &       66\% &       4,04 \\
\hline
 urf3\_155   &         10 &      26468 &     132340 &     331578 &    2121808 &    2038584 &    1933372 & {\bf 596368} &     596371 &       72\% &       4,51 \\
\hline
 urf5\_158   &          9 &      10276 &      51380 &     114784 &     740084 &     706412 &     667484 &     208709 & {\bf 208706} &       72\% &       4,06 \\
\hline
 urf6\_160   &         15 &      10740 &      53700 &     239034 &    1487904 &    1478080 &    1334916 &     320412 & {\bf 320409} &       78\% &       5,97 \\
\hline

\end{tabular}
\end{changemargin}
\vspace{.6cm}
\end{table*}
}

As can be seen, decomposing reversible circuits to have NNC-optimal quantum circuits is costly. Using the naive method, the quantum cost increases by one order of magnitude on average. However, this can be significantly improved if macros or reordering are applied. 
Even if reordering may worsen the results in some few cases (e.g. for local reordering in 0410184\_169 or add64\_184), in total this leads to an improvement of 50\% on average -- in the best case 83\% improvement was observed. Furthermore, since all execution times are negligible, it is feasible to run all decomposition methods and afterwards choose the best circuit as marked bold. As a result, NNC-optimal circuits can be synthesized
with a moderate increase of quantum cost.

\section{Conclusions and Future Work}  \label {sec:conclusion}

In this work, we examined the synthesis of reversible functions with respect to cost metrics beyond gate count and quantum cost. While most of
the previous synthesis approaches only take gate count and quantum cost into account, we showed how the evaluation differs
if other realistic metrics are applied. Furthermore,
by considering NNC as a cost metric for the linear nearest neighbor quantum architecture,
we illustrated how the available synthesis flows can be modified to produce NNC-optimal circuits.
Improvements were suggested that
reduce the quantum cost of the results by up to 83\% (56\% on average) without affecting the NNC-optimality.

For future work, synthesis approaches should be modified with respect to further cost metrics. In particular, quantum-related metrics concurrency and gate distribution have not been addressed by synthesis approaches so far. Therefore, consideration of these metrics is seen as our natural next step. For this purpose, applying exact approaches and line reordering methods similar to those proposed in the paper may be useful.

\section*{Acknowledgment}
Parts of this research work has been supported by the German Research Foundation (DFG) (DR 287/20-1).

\end{document}